\documentclass[conference]{IEEEtran}
\usepackage{cite}
\usepackage{graphicx}
\usepackage{graphics}
\usepackage{epsfig}
\usepackage{epstopdf}
\usepackage[tight,footnotesize]{subfigure}
\graphicspath{{../}}
\DeclareGraphicsExtensions{.eps,.pdf,.jpeg,.png,.jpg}
\usepackage[cmex10]{amsmath}
\usepackage{algorithmic}
\usepackage{array}
\usepackage{mdwmath}
\usepackage{mdwtab}
\usepackage{cases}
\usepackage{eqparbox}
\usepackage{caption}
\usepackage{amsthm}
\usepackage{algorithm}
\usepackage{algorithmic}
\usepackage{stfloats}
\usepackage{url}
\usepackage{setspace}
\usepackage{eqparbox}
\usepackage{balance} 
\balance 
\usepackage{multicol}
\usepackage{multirow}
\usepackage{amsfonts}
\usepackage{amssymb}
\usepackage{indentfirst}

\newtheorem{theorem}{Theorem}

\newtheorem{lemma}{Lemma}

\newtheorem{example}{Example}

\begin{document}
\title{Concurrent Regenerating Codes and Scalable Application in Network Storage
}

\author{\IEEEauthorblockN{Huayu Zhang\IEEEauthorrefmark{1},
Hui Li\IEEEauthorrefmark{1},
Hanxu Hou \IEEEauthorrefmark{1}\IEEEauthorrefmark{2},
K. W. Shum\IEEEauthorrefmark{2} and ShuoYen Robert Li\IEEEauthorrefmark{2}}
\IEEEauthorblockA{\IEEEauthorrefmark{1}Shenzhen Graduate School, Peking University, China\\
Email: {zhanghuayu,lih64}@pkusz.edu.cn}
\IEEEauthorblockA{\IEEEauthorrefmark{2}Institute of Network Coding, the Chinese University of Hong Kong, Shatin, Hong Kong}}
%
\maketitle
\IEEEpeerreviewmaketitle
\begin{abstract}
To recover simultaneous multiple failures in erasure coded storage systems, Patrick Lee \textit{et al} introduce concurrent repair based minimal storage regenerating codes to reduce repair traffic. The architecture of this approach is simpler and more practical than that of the cooperative mechanism in non-fully distributed environment, hence this paper unifies such class of regenerating codes as concurrent regenerating codes and further studies its characteristics by analyzing cut-based information flow graph in the multiple-node recovery model. We present a general storage-bandwidth tradeoff and give closed-form expressions for the points on the curve, including concurrent repair mechanism based on minimal bandwidth regenerating codes. We show that the general concurrent regenerating codes can be constructed by reforming the existing single-node regenerating codes or multiple-node cooperative regenerating codes. Moreover, a connection to strong-MDS is also analyzed. 

On the other respect, the application of RGC is hardly limited to "repairing". It is of great significance for "scaling", a scenario where we need to increase(decrease) nodes to upgrade(degrade) redundancy and reliability. Thus, by clarifying the similarities and differences, we integrate them into a unified model to adjust to the dynamic storage network.

\end{abstract}


\section{Introduction} \label{1}
With the scale of distributed storage system (DSS) growing, component (e.g., disk, server or rack) failures due to various reasons become normal events, driving redundant strategy to provide fault-tolerance reliability. Replication is a simple way used in most systems, such as GFS \cite{ghemawat2003google}, HDFS \cite{borthakur2008hdfs} and S3 \cite{garfinkel2007evaluation}, where triple copies are dispersed across different nodes. In general, $m+1$ copies are needed to tolerant $m$ failures, which is a high cost of storage. 
Due to the lower storage cost but higher reliability than simple replication, erasure codes become popular in new generation DSS, such as Cleversafe \cite{cleversafe2008paradigm}, GFS2 \cite{whitehouse2007gfs2} and HDFS-RAID \cite{borthakur2010hdfs}.
However, it is a disadvantage of erasure codes to repair failure nodes in a precious bandwidth environment. For instance, we divide an original file, with a size of $B$, into $k$ pieces, encode them into $n$ pieces and distribute the coded pieces to $n$ storage nodes, each of which stores $\frac{B}{k}$ data. Although the \textit{(n,k)} maximum distance separable (MDS) property guarantees that any $k$ out of $n$ nodes can reconstruct the original file, we must download the whole file to recover mere one piece in a replacement node (newcomer) to maintain the system in the same state when there is one node failed.

To reduce the repair traffic, network coding \cite{li2003linear} is applied to erasure codes in the premise of keeping MDS property. Reconsider the above example, the $\frac{B}{k}$ data in each piece are further divided into $d-k+1$ strips, a linear combination of these strips is downloaded from $d$ active nodes (helpers) to recover the lost piece. To some extent, the total repair bandwidth can be reduced to $\frac{Bd}{k(d-k+1)}$ for $d>k$.
Using network coding, Dimakis \textit{et al} \cite{dimakis2007benefits} \cite{dimakis2010network} cast the storage problem as a network multicast problem and clarified the tradeoff between storage and bandwidth. They proposed regenerating codes (RGC) based on the points of the tradeoff curve. The famous codes are the minimal storage regenerating (MSR) codes and minimal bandwidth regenerating (MBR) codes that based on the two extreme points on the curve. Subsequently, varieties of construction of the codes are proposed \cite{RashmiShah-75} \cite{SuhRamchandran-78} \cite{wu2010existence} \cite{5773063}. A survey launched by Dimakis \cite{dimakis2011survey} summarizes the rapid development of RGC. Nevertheless, the researches mainly center on single failure. 

In some situation, there might be multiple failures, say $t>1$ failures. To repair multiple failures, the repair traffic is still rather heavy in the way of conventional process. Thus, it is quite natural to turn to RGC for help. The simplest method is to repair them one by one by using RGC, which can be independently optimal in each recovery. Yet the cumulative traffic is not minimal. Wang \textit{et al.} \cite{wang2010mfr} improved this method by allowing a newcomer to connect to both of the helpers and repaired nodes to complete the whole recovery. It is a multi-loss flexible recovery (MFR) mechanism focusing on minimal storage point.
Nevertheless, it is suboptimal in contrast with cooperative RGC \cite{5402494} \cite{5962548}, which can minimize this traffic by further mutually exchanging information among $t$ repaired nodes after each node downloads data from original helpers. Both the original traffic from helpers to newcomers and the coordinate traffic among repaired nodes are taken into consideration. Kermarrec \textit{et al.} \cite{5978920} went beyond these work and proposed adaptive regenerating codes. Shum and Hu \cite{6565355} exhaustively presented the storage-bandwidth tradeoff, on the curve of which are the two extreme points that respectively correspond to minimal-storage cooperative regenerating (MSCR) codes and minimal bandwidth cooperative regenerating (MBCR) codes. Explicit constructions of MBCR are proposed in \cite{6566803} and that of MSCR with different parameters in \cite{6620465} \cite{6847953}. 

However, the disadvantages of cooperative RGC mainly are:
\begin{enumerate}
\item[1] In each regeneration period, every newcomer needs to communicate with $d$ helpers in the first phase and $t-1$ other newcomers in the second phase. We treat the information channel between two nodes as link. Logically, the system need to maintain totally about $td+t(t-1)/2$ links. 
\item[2] The data sent out to other newcomers at the second phase cannot be generated until the newcomer receives all the data from $d$ helpers. Meanwhile, the final recovered data cannot be generated until the newcomer receives all the data from other $t-1$ newcomers. These processes need high synchronization and consistency.
\item[3] For this mechanism is highly coupled, one minor error in a phase would lead to the failure of the whole repair. 
\end{enumerate} 

The application of cooperative RGC requires a fully decentralized network. The way of "teamwork" for repair highly depends on the cooperation of other members, then the codes have limitations (e.g., robustness, consistency and management problems) in some environments, such as the centralized-management datacenters. Therefore, we need a more appropriate method for repairing multiple failures in non-fully distributed network.

In fact, Li \textit{et al.} \cite{DBLP:journals/corr/abs-1302-3344} \cite{6880379} developed a distributed storage system, the CORE, based on MSR codes. It introduces a concurrent repair framework for both single and multiple failures.  The repair process in CORE is that the engine collects data from $d$ helpers, regenerates all the lost data and then disperse them into newcomers. Similar to conventional erasure coded system, CORE only considers the traffic that help to recover the whole lost data and ignore the retransmitted traffic of recovered data to distributed nodes. The framework reserves the advantages of erasure codes so that it can be easily accepted and deployed in the mainstream systems, such as HDFS-RAID \cite{you2014repairing}.

Nevertheless, the authors only focus on the minimal storage point and prefer the system design to theoretical analysis. In this paper, we further study this mechanism, analyze its cut-based information flow graph and give the closed-form expressions for the points on the storage-bandwidth tradeoff curve.
Moreover, CORE used \textit{virtual nodes} to prove the possibility of multiple repair and showed the bad failure pattern that can not be repaired. We disproves the existence of bad failure pattern.
Meanwhile, we build a connection to strong-MDS properties proposed in \cite{5402494} when $t=k$.

On the other respect, DSS provides both huge storage volume for magnanimity information and parallel services for billions of intensive accesses \cite{fan2009diskreduce}. Apart from ensuring no information lost, the system needs to adjust the redundancy of data according to the change of workload. For example, when new software is released, software company might distribute as many copies as possible to cope with the outburst downloading, regardless of using cache or disk. When the hot degree of the information goes down, it merely keep few ones for the normal downloading. Thus, flexibility is of the same importance as reliability for content delivery network (CDN) and information centric network (ICN) \cite{6563278}. Based on the same perspective that replication means more resource and energy consuming, we study the role of erasure codes (EC) in scalability. For instance, we upgrade $(n,k)$ codes into $(n+3,k)$ codes to store more data pieces. Thus, similar problems crop up because of large upgrade bandwidth if the original file is invalid. Since the flexibility of multicast in turn implies the scalability of storage, we extract potential value of RGC and initially represent its fetching characteristics in scalability. The second result in this paper is that we extend the functional RGC for upgrade and clarify the difference between repair and upgrade.


The remainders of this paper are organized as follows: in section \ref{2} we present preliminary background and related work about failure repair in erasure codes DSS. In section \ref{3} we analyze cut-based information flow graph in a concurrent model and give closed-form expressions for the points on the tradeoff curve. We also present how to construct concurrent regenerating codes by using existing approaches and prove the pervasive strong-MDS property. The second contribution of this paper is presented in section \ref{4}. We give our conclusion in section \ref{5}. We summer the notation in table \ref{notations}.
\begin{table}[ht] 
\caption{Summary of key notations} 
\small{
\begin{tabular}{ll}
\hline
\textbf{Notation} & \textbf{Meaning}\\
\hline
$B$ & The size of the source file.\\
$n$ & The total number of storage nodes. \\
$k$ & The minimal number of nodes for reconstruction.\\
$d$ & The total number of helpers.\\
$\mathbf{h}$ & The $d$-length capacity vector for upgrade.\\
$k^*$ & The flexible $k$ for strong-MDS property.\\
$\mathbf{c}^*$ & The $k^*$-length capacity vector for reconstruction.\\
$t$ & The number of nodes repair concurrently.\\
$s$ & The number of extended nodes for upgrading \\
    & $(n,k)$ to $(n+s,k)$ MDS codes.\\
$\alpha$ & Storage per node.\\
$\beta$ & Repair/upgrade bandwidth from each helper.\\
$\gamma$ & The total repair/upgrade bandwidth.\\
$\lambda$ & The number of links built in a repair/upgrade \\
    & scenario.\\
\hline
\end{tabular}}
\end{table} \label{notations}
\section{Preliminary Background And Related Work} \label{2}

\subsection{Regenerating Code for Single Failure} \label{2a}
\begin{table*} \centering
\renewcommand\arraystretch{1.5}
\caption{Cut-set bound of information flow graph in erasure coded system}\label{cutsetresult}
\begin{tabular}{|c|c|c|}
\hline
type& closed-form expression & reference \\
\hline
single repair & $\sum_{i=0}^{\min (d,k)-1} \min \{(d-i)\beta, \alpha\}$ & \cite{dimakis2010network} \\
\hline
cooperative multiple repair & $\min_{\mathbf{u} \in P}(\sum_{i=0}^{g-1}u_i\min(\alpha,(d-\sum_{j=0}^{i-1}u_j)\beta_1+(t-u_i)\beta_2)$ & \cite{5962548}\cite{6565355} \\
\hline
concurrent multiple repair & $\min_{\mathbf{u} \in P}(\sum_{i=0}^{g-1}\min(u_i\alpha,(d-\sum_{j=0}^{i-1}u_j)\beta)$ & this paper \\
\hline
\end{tabular}
\end{table*} 
By casting the storage problem as a multicast communication problem, Dimakis \textit{et al.} \cite{dimakis2007benefits}\cite{dimakis2010network} analyzed the cut-based information flow grow graph for single failure.
\begin{lemma}\label{mincutsingle}
\textbf{(Mincuts of Information Flow Graphs in Single Repair Model \cite{dimakis2010network})}: Consider the information flow graph $\mathcal{G}(n, k, d, \alpha, \beta)$ formed by $n$ initial nodes connecting to a virtual source (VS) and obtaining $\alpha$ units data. The additional nodes join the graph by  connecting to $d$ existing nodes, obtaining $\beta$ units from each and storing $\alpha$ units data per node. A data collector (DC) can connect to arbitrary $k$-node subset of $\mathcal{G}$ to reconstruct the original file, which must satisfy:
\begin{equation}
mincut(VS,DC) \ge \sum_{i=0}^{\min (d,k)-1} \min \{(d-i)\beta, \alpha\}
\end{equation}
\end{lemma}
We call the \textit{capability} of graph $\mathcal{G}^*$ \[C(\mathcal{G}^*)\triangleq \sum_{i=0}^{\min (d,k)-1} \min \{(d-i)\beta, \alpha\}.\]
To guarantee that the collector can reconstruct the origin file, the necessary condition is $C(\mathcal{G}^*) \ge B$, which drives the tradeoff between storage $\alpha$ and repair bandwidth $\gamma = d\beta$. 
Thus, the codes that can achieve every point on this optimal tradeoff curve are called regenerating codes (RGC). Specially, the two extreme points of this bound are greatly practical in real systems, namely, minimum storage regenerating (MSR) codes:
\begin{equation} \label{MSR}
(\alpha_{MSR}, \gamma_{MSR})=(\frac{B}{k}, \frac{Bd}{k(d - k+1)})
\end{equation}
and minimum bandwidth regenerating (MBR) codes:
\begin{equation} \label{MBR}
(\alpha_{MBR},\gamma_{MBR})=(\frac{2Bd}{k(2d - k + 1)}, \frac{2Bd}{k(2d - k + 1)}).
\end{equation}
In repair situation, the range of $d$ is subject to
\begin{equation}\label{d-range}
k \le d \le n-1
\end{equation}
when $d =n-1$, repair bandwidth of both MSR and MBR gets minimal.

According to whether the repaired information is exactly the pre-lost data, there are mainly three repair models based on the two codes: \textit{exact repair}, \textit{functional repair} and \textit{exact repair of systematic parts} \cite{dimakis2011survey}. Functional repair mainly uses random linear network coding technology under proper finite filed \cite{wu2010existence}.  In addition to matrix product (PM) \cite{RashmiShah-75} construction for both exact MSR ($n\ge 2k-1$) and exact MSR , there are mainly interference alignment (IA) \cite{SuhRamchandran-78} construction for MSR ($n\ge 2k-1$) and repair-by-transfer \cite{6062413} mechanism for MBR. 

Note that the interior points on the tradeoff are impossible for exact repair \cite{6062413} and hardly significant for functional repair because of involving more variable parameters, only MSR and MBR are taken into account in practical system.
\subsection{Cooperative Regenerating Codes for Multiple Failure} \label{2b}
Due to the salient feature of minimizing repair traffic for single failure, it is natural to apply RGC step by step to repair multiple failure. However, although the traffic is optimal for each step, the cumulative traffic is not optimal. To improve this strategy, a MFR \cite{wang2010mfr} mechanism is proposed on the minimal storage point. The $t$ replacement nodes are denoted as sequence $y_1, \dots, y_t$ and then $y_i$ can download data from any available nodes, including active original nodes and other repaired newcomers, e.g., $y_j, j<i$. Suppose there are $d_i$ helpers for repairing $y_i$, the lower bound of the total repair traffic is
\begin{equation}\label{MFR}
\sum_{i=1}^{t} \frac{Bd_i}{k(d_i-k+1)}
\end{equation}
In the practical code construction, we split the original file into $mk$ packets, where $m$ is the Least Common Multiple (LCM), namely $m=LCM(d-k+1, \dots, d-k+t)$. 
Whereas, MFR inefficiently uses the help of newcomers since the later repaired newcomers have no contribution to the former ones, e.g., $y_j$ cannot download data from $y_i$ for $j > i$. Hence, the optimal way is to sufficiently exchange information between these newcomers, which is just the achievement of the cooperative regenerating codes \cite{5402494} \cite{5962548} \cite{5978920} \cite{6565355}. 
In the first phase, each newcomer connects to $d$ active helpers and downloads $\beta_1$ unit data from each helper. In the second phase, each newcomer exchanges $\beta_2$ unit with other $t-1$ newcomers.

The cut analysis of information flow graph for cooperative repair drives the following lemma:
\begin{lemma}\label{mincutmultiple}
\textbf{(Mincuts of Information Flow Graphs in Cooperative Multiple Repair Model \cite{5978920})}: Consider the information flow graph $\mathcal{G}(n, k, d, \alpha, \beta, t)$ formed by $n$ initial nodes connecting to a virtual source (VS) and obtain $\alpha$ units data. The additional nodes join the graph in group of $t$ nodes by connecting to $d$ existing nodes, obtaining $\beta_1$ units from each existing node and $\beta_2$ units from other $t-1$ joining nodes and storing $\alpha$ units data per node. A data collector (DC) can connect to arbitrary $k$-node subset of $\mathcal{G}$ to reconstruct the original file, which must satisfy:
\begin{gather}
mincut(VS,DC) \ge \notag \\ 
\min_{\mathbf{u} \in P}(\sum_{i=0}^{g-1}u_i\min(\alpha,(d-\sum_{j=0}^{i-1}u_j)\beta_1+(t-u_i)\beta_2)
\end{gather}
where $\mathbf{u}=[u_i]_{1 \times g}$, $\sum_{i=0}^{g-1}u_i=k, 1 \le u_i \le t$, $\lceil \frac{k}{t} \rceil \le g \le k$, and $P$ is set of all possible $\mathbf{u}$. Set $\sum_{j=0}^{-1}u_j = 0$.
\end{lemma}
In \cite{5962548}\cite{6565355}, the \textit{capability} of the flow graph is simplified into 
\[C(\mathcal{G}^*)\triangleq u_0\alpha + \sum_{i=0}^{g-1}u_i((d-\sum_{j=0}^{i-1}u_j)\beta_1+(t-u_i)\beta_2).\]
Then, for $i = 0,\dots, g-1$, the storage-bandwidth tradeoff is
\begin{equation}
\alpha=\frac{B(d-k+t(i+1))}{D_i}
\end{equation}
\begin{equation}
\gamma=\frac{Bt(d+t-1)}{D_i}
\end{equation}
where
\begin{equation} \label{Di}
D_i = k(d+t(i+1)-k)-\frac{i(i+1)t^2}{2}
\end{equation}

The two extreme points on this tradeoff curve are respectively minimal storage cooperative regenerating (MSCR) point and  minimal bandwidth cooperative regenerating (MBCR) point.
For MSCR,
\begin{equation} \label{MSCR}
(\alpha_{MSCR}, \gamma_{MSCR})=(\frac{B}{k}, \frac{Bt(d+t-1)}{k(d - k+t)})
\end{equation}
For MBCR,
\begin{equation} \label{MBCR}
(\alpha_{MSCR}, \gamma_{MSCR})=(\frac{B(2d+t-1)}{k(2d+t-k)}, \frac{Bt(2d+t-1)}{k(2d+t-k)})
\end{equation}
In the same literature \cite{6565355}, Shum, K.W. \textit{et al.} proposed the functional construction of the two codes based on linear network coding. 
The explicit product-matrix based constructions of exact MBCR and exact MSCR are proposed in \cite{6566803} and in \cite{6847953} respectively. Similarly, Chen \textit{et al.} \cite{6620465} introduce the interference alignment to MSCR.

Cooperative repair can be adopted in a regular distributed and robust network. However, due to the shortcomings would complicate the management, the benefits of this approach may be outweighed by its drawbacks in the real systems, like the huge computer cluster of GFS2.
\subsection{Concurrent Repair Mechanism}\label{2c}
\begin{figure*}[!t] \centering
	\includegraphics[width=6.5in]{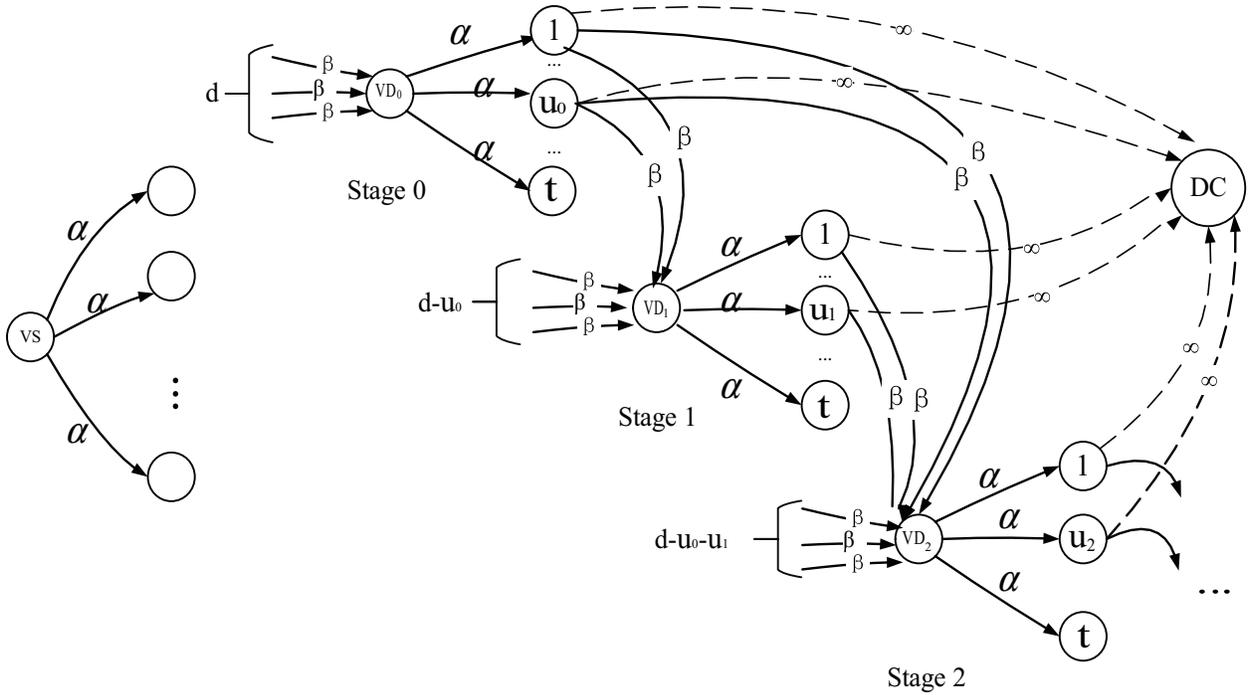}
	\caption{The information graph $\mathbf{G}(n, k, d, \alpha, \beta, t)$ of concurrent repair mechanism }
	\label{fig_mincut}
\end{figure*}
Li \textit{et al.} separately realize related application in NCFS \cite{hu2011ncfs}, NCCloud \cite{hu2012nccloud} and CORE \cite{DBLP:journals/corr/abs-1302-3344}. NCFS presented its bandwidth saving for single repair under same tolerant nodes as Reed-Solomon codes, NCCloud shifts NCFS in cloud storage and shows functional MSR can minimize both repair bandwidth and I/O in archive (cold) data storage. Based on the previous work of NCFS and NCCloud, CORE supports both single and multiple failure repair. In CORE, the authors proposed a concurrent repair mechanism . When $t$ nodes fail, an engine starts to download $\gamma$ data from $d$ available nodes and then dispatch recovered data to newcomers, the lower bound of repair bandwidth is:
\begin{equation} \label{PatreekMSR}
\gamma =
\begin{cases}
\frac{Btd}{k(d+t-k)} & t<k, \\
B & t \ge k.
\end{cases}
\end{equation}
where $\alpha =\frac{B}{k}$. For conventional erasure coded storage system, once the number of failure reaches a certain threshold, a daemon starts to collect $B$ data, re-encode $t$ pieces and then dispatch them into $t$ replacement nodes. In this way, only the "collect" bandwidth is treated as the actual network cost. So the traffic from the engine to the newcomers has no contribution to $\gamma$ in this framework. Namely, if the engine runs in one newcomer, the traffic among $t$ newcomers (i.e., the bandwidth of the second phase in cooperative repair) is ignored.  

Apart from the minimal storage point, a closer study of concurrent repair mechanism is necessary because of its practical applications. We present a similar storage-bandwidth tradeoff based on the cut analysis of information flow graph.   

\section{Concurrent Regenerating Codes} \label{3}
\subsection{A Cut-set Bound Analysis} \label{3a}
We analyze the cut of the information flow graph and find a lower bound of bandwidth for generating multiple nodes in the concurrent repair as CORE. By referring the procedure of literature \cite{5978920}, we have the following theory.
 
\begin{theorem} \label{generalMinCut}
Denote the information flow graph as $\mathbf{G}(n, k, d, \alpha, \beta, t)$ formed by $n$ initial nodes connecting to a virtual source (VS) and obtaining $\alpha$ units data. There are $t$ newcomers connecting to a virtual daemon (VD) and storing $\alpha$ units data for each. The virtual daemon connect to $d$ existing nodes and obtaining $\beta$ units from each. Any data connector (DC) can connect to any $k$-subset of nodes of G must satisfy
\begin{equation} \label{equmincut}
mincut(VS,DC) \ge \min_{\mathbf{u} \in P}(\sum_{i=0}^{g-1}\min(u_i\alpha,(d-\sum_{j=0}^{i-1}u_j)\beta)
\end{equation}
with $P = \{\mathbf{u}=[u_i]_{1 \times g} |\sum_{i=0}^{g-1}u_i=k, 1 \le u_i \le t \}$.
\end{theorem} 

\begin{proof}
In the graph illustrated in figure \ref{fig_mincut}, there are initially $n$ nodes labeled from $1$ to $n$ connecting to virtual source VS with each edge capacity of $\alpha$. At time stage $i$, a virtual daemon labeled as $VD_i$ is used to generate $t$ newcomers labeled as $x_1^i, \dots, x_t^i$, each of which connects to $VD_i$ with edge capacity of $\alpha$. We call these nodes the children of $VD_i$.
 
To recover the whole file, a data collector labeled as $DC$ collects data from any $k$-node subset of the whole active nodes, denoted as $U$. Suppose DC selects $u_i$ children of $VD_{i}$ (the corresponding set is $U_{i}$) and there are $g$ group of children to construct $U$. Let $I = \left\lbrace 0,\dots, g-1\right\rbrace$ denote the set of group labels. Then, a recovery scenario is defined as a sequence $\mathbf{u} = [u_i]$, $i \in I$, $1 \le \left|U_i \right|=u_i \le t$ and $\sum_{i=1}^{g}u_i=k$, $\lceil \frac{k}{t} \rceil \le g \le k$, $U=\bigcup_{i=1}^gU_i$. Also, let $P$ be the set of all the recovery scenarios.


For each regenerating process, the daemon node $VD_i$ connects to $d$ helpers. One part of the helpers is composed of $\sum_{j=0}^{i-1} v_j$ newcomers, where $v_j$ is the number of selected children of $VD_j, j=\left\lbrace0, \dots i-1 \right\rbrace$, and the other part is composed of $d-\sum_{j=0}^{i-1} v_j$ ($>0$) initial nodes. 
Thus, sequence $ v_0,\dots, v_{g-1} $ constructs a regenerating scenario $\mathbf{v}=[v_i], i \in I, 1 \le v_i \le t$. Let $V_i$ be the corresponding set of $v_i$ helpers, then $V=\bigcup V_i$. Also, let $Q$ be the set of such scenarios. Since the number of generating stage is unlimited while recovery scenario is constrained by $g$ groups, only considering the same number of groups for regenerating scenarios can be enough analyzing the min-cut. In other words, we have
\begin{equation} \label{QincludeP}
\forall \mathbf{u} \in P, \exists \mathbf{v} \in Q \to \mathbf{u}=\mathbf{v}.
\end{equation}

Note that in both scenarios we have $j < i$, for nodes of the $i$-th stage cannot depend on nodes considered at $j$-th stage with $j > i$.

Moreover, we denote the intersection of helper set and generating set on the $i$-th stage as $L_{i}=U_{i} \bigcap V_{i}$ with size of $\ell_{i}=\left|L_i \right|$.
 
For a regenerating scenario $\mathbf{v} \in Q$ and a recovery scenario $\mathbf{u} \in P$, we consider the min-cuts ($A,\bar{A}$) with $VS \in A$ and $U \subset \bar{A}$.

Let $\mathcal{C}$ denote the edges in the cut. Firstly, consider the  $\mathcal{C}$ on the $i$-th stage, denoted as $\mathcal{C}_{i}$. 
\begin{figure}[t]
\begin{center}
\subfigure[]{\includegraphics[width=0.45\textwidth]{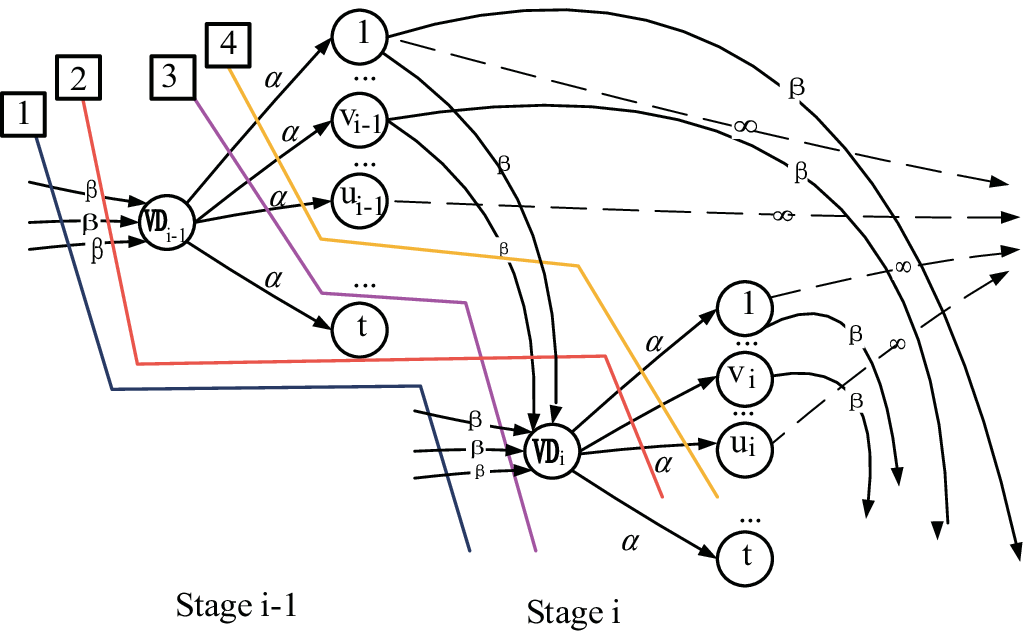}}
\\
\subfigure[]{\includegraphics[width=0.45\textwidth]{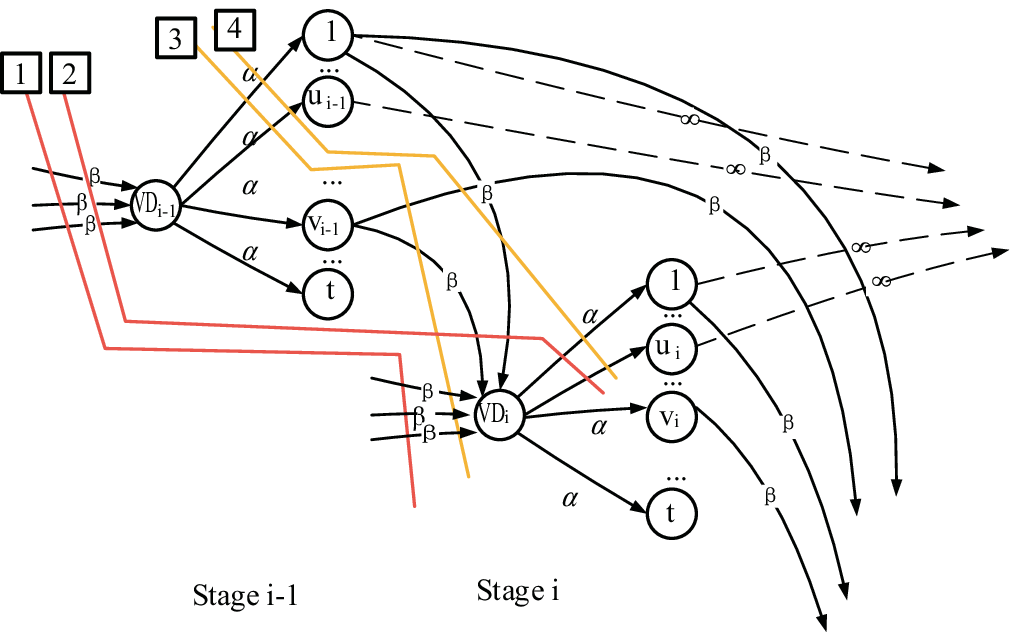}}
\caption{\small{The two cases of cut at stage $i$. (a) case ($i$) (b) case ($ii$).}}\label{fig_cutstagei}
\end{center}
\end{figure}
If $VD_i \in A$, then the edges $VD_i \to x_j^i \in \mathcal{C}_{i}$, where $j \in U_i$, each with capacity of $\alpha$. 

If $VD_i \in \bar{A}$, there are two cases to be considered.
\begin{enumerate}
\item[($\romannumeral1$)]As illustrated in figure \ref{fig_cutstagei}(a), if $\forall j < i$, $VD_{j} \in \bar{A}$, then the $d-\sum_{j=0}^{i-1}v_{j}$ edges from initial node are in $\mathcal{C}_{i}$, each with capacity of $\beta$.
\item[($\romannumeral2$)]In figure \ref{fig_cutstagei}(b), if $\exists S \subset \left\lbrace 0,\dots,i-1 \right\rbrace$,$VD_{j:j \in S} \in A$, both the $d-\sum_{j=0}^{i-1}v_{j}$ edges from initial node and $\sum_{ j \in S }(v_{j}-\ell_{j})$ edges from $\bigcup_{j \in S} (V_{j}/L_{j})$ are in $\mathcal{C}$, each with capacity of $\beta$, since there are edges from children of $VD_{s} \in S$ to $VD_i$.
\end{enumerate}
Therefore, the minimal value of $\mathcal{C}_{i}$ is
\begin{equation} \label{minCi}
\min \mathcal{C}_{i}=\min (u_i\alpha, c_i)
\end{equation}
where
\begin{numcases}
{c_i=}
(d-\sum_{j=0}^{i-1}v_i)\beta, & case ($\romannumeral1$)\\
(d-\sum_{j=0}^{i-1}v_i+ \bigcup_{j \in S}(v_{j}-\ell_{j}))\beta, & case ($\romannumeral2$)
\end{numcases}
Since $\forall j \ge 0, \ell_{j} \le v_{j}$ and when $V_{j} \subset U_{j}$ the cut of case ($\romannumeral2$) reach the same minimal value as case ($\romannumeral2$), then $v_{i} \le u_{i}$ on the $i$-th stage.
Refer to (\ref{QincludeP}), we need to maximize $v_{i}$ to minimize $c_i$ for all the regenerating scenarios under a recovery scenario. Thus if $u_i=v_i$ or $U_{j} = V_{j}$, (\ref{minCi}) can be rewritten as
\begin{equation}
\min \mathcal{C}_{i}=\min (u_i\alpha, (d-\sum_{j=0}^{i-1}u_i)\beta)).
\end{equation}
Finally, we have min-cut of $\mathcal{C}$ under certain recovery scenario
\begin{equation} \label{minC}
\min \mathcal{C} = \sum_{i=0}^{g-1}\min \mathcal{C}_{i}=\sum_{i=0}^{g-1}\min (u_i\alpha, (d-\sum_{j=0}^{i-1}u_i)\beta))
\end{equation}
If $\forall \mathbf{u} \in P$, the minimal value of (\ref{minC}) is just the min-cut of (VS, DC). Hence the claim follows.
\end{proof}
\subsection{The Storage-bandwidth Tradeoff} \label{3b}
Next, we need to find the similar tradeoff between storage and bandwidth. With the same assumption as \cite{dimakis2010network}, an valid generation must satisfy
\begin{equation} \label{mincutbound}
\widehat{C} \triangleq \min_{\mathbf{u} \in P}\sum_{i=0}^{g-1}\min(u_i\alpha,(d-\sum_{j=0}^{i-1}u_j)\beta) \ge B.
\end{equation}
where $\widehat{C}$ denotes the \textit{capacity} of $\mathbf{G}(n, k, d, \alpha, \beta, t)$. 

Note that single node repair is a special situation of (\ref{mincutbound}) when $t = 1$, for there is only one recovery scenario, $\mathbf{u}=[1]_{1\times k}$ . Hereinafter, we assume $t>1$ if there is no special specification. 
Firstly, it is easy to compute (\ref{minC}) under a certain $\mathbf{u}$ while hard to find the minimal one under all the $\mathbf{u}$. Since the total number of recovery scenarios grows exponentially as $t$ growing, which is , in fact, an integer composition problem that satisfies
$\sum_{i=0}^{g-1}u_i = k$
and subjects to $1 \le u_i \le t$. 
As a example in table \ref{table_intpartitions}, we enumerate all the possible compositions of $\mathbf{u}$ when $k = 7$ and $t=3$. To be simple, the second column of table only illustrates all the partitions (composition is an order matter of partition) of $\mathbf{u}$ and the third column list the number of corresponding compositions. Therefore, it is unpractical to check every compositions. Then we have the following theorem:
\begin{table} \centering 
\caption{partitions of integer k=7 when t=3} \label{table_intpartitions}
\begin{tabular}{|c|c|c|}
\hline
$g$ & $\mathbf{u}$ & number of cases \\
\hline
7 & [1,1,1,1,1,1,1] & 1\\
\hline
6 & [2,1,1,1,1,1] & 6\\
\hline
5 & [3,1,1,1,1];[2,2,1,1,1] & 15\\
\hline
4 & [2,2,2,1];[3,2,1,1] & 16\\
\hline
3 & [2,2,3];[3,3,1] & 6\\
\hline
\end{tabular}
\end{table}
\begin{figure}[t]
\begin{center}
\subfigure[]{\includegraphics[width=0.45\textwidth]{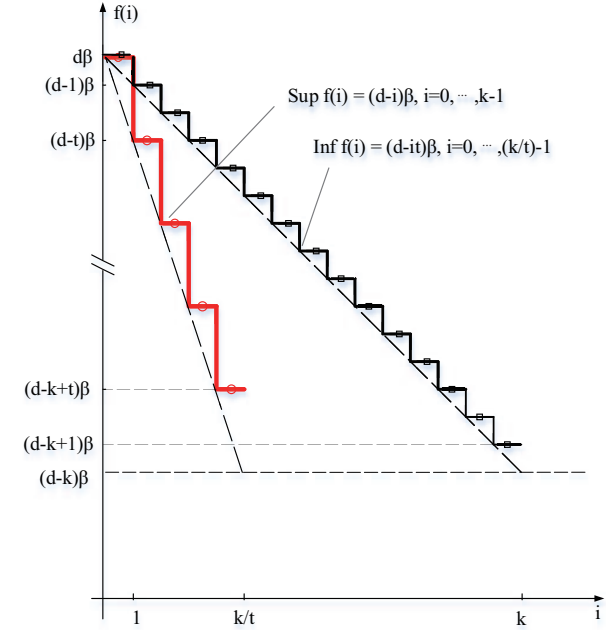}}
\subfigure[]{\includegraphics[width=0.45\textwidth]{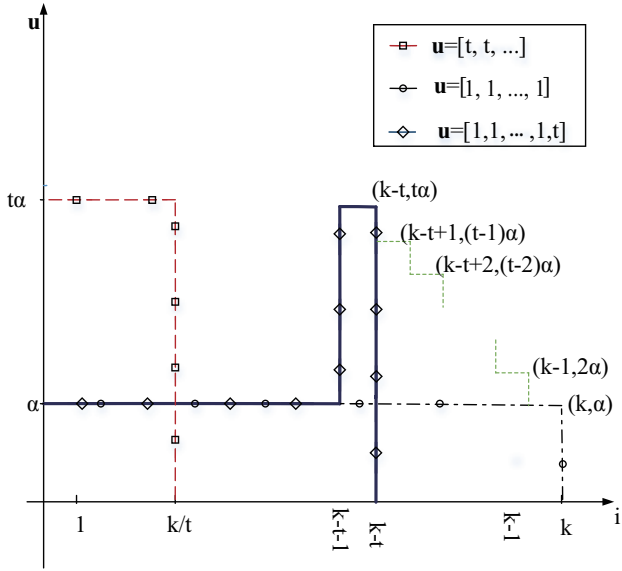}}
\caption{\small{The stair function $f(i)$ and $g(i)$ under different $\mathbf{u}$. (a) The lower and upper bound of $f(i)$. (b) three examples of $g(i)$ }}\label{fig_function_fg}
\end{center}
\end{figure}
\begin{theorem}
the \textit{capacity} of $\mathbf{G}(n, k, d, \alpha, \beta, t)$ is minimal if and only if the composition of $k$ is
\begin{equation} \label{valueof-u}
\mathbf{u}=
\begin{cases}
[t,\dots,t]_{1 \times g} & t \mid k,\\
[k-\lfloor \frac{k}{t}\rfloor t,t,\dots,t]_{1\times g} & t \nmid k.
\end{cases}
\end{equation}
\end{theorem}
\begin{proof}
Consider the function $f(i)=(d-\sum_{j=0}^{i-1}u_{j})\beta$ and $g(i) = u_i\alpha$, $0 \le i \le g-1$. Define set $Y = \left\lbrace u_i | g(i) \ge f(i), i \in I  \right\rbrace$ and $\bar{Y} = \left\lbrace u_i | g(i) < f(i), i \in I   \right\rbrace$. 
Suppose $t \mid k$, when $\mathbf{u}$ take the value of (\ref{valueof-u}), we can obtain the lower bound of $f(i)$ 
\begin{equation*}
\inf f(i) = (d-it)\beta 
\end{equation*}
when $\mathbf{u}=[1]_{1\times k}$, the upper bound is \[\sup f(i) = (d-i)\beta.\]
As illustrated in figure \ref{fig_function_fg}(a), for all the other kinds of $\mathbf{u}$, $f(i)$ locates between the two bound. 
Here we use a stair function to describe $f(i)$ since the area closed by $i$-axis, vertical axis, $f(i)$ and $i=g-1$ equals to $\widehat{C}$ when $\bar{Y} = \emptyset$. Note that $\inf f(i)$ make the area minimal.
On the other respect, function $g(i) = u_i\alpha$ with $\mathbf{u} = [t,t,\dots]$ and $\mathbf{u} = [1,1,\dots]$ is depicted in figure \ref{fig_function_fg}(b). When $Y = \emptyset$, the similar closed area reflects the value of $\widehat{C}$. Note that the area closed purely by arbitrary $g(i)$ equals to $k\alpha$.
 
We assume that $\beta$ is fixed and $\alpha$ is variable and merge the two function into the same coordinate system. Then $\widehat{C}$ equals to the intersection area closed by $i$-axis, vertical axis, $\min (g(i),f(i))$ and $i=g-1$ for all the situations of $Y$. Because $\inf f(i)$ make the minimal area, thus we only scale $g(i)$ with the same $\mathbf{u}$ as $\inf f(i)$ to find minimal $\widehat{C}$ among all $\mathbf{u} \in P$.

Suppose $t \nmid k$ and $k-\lfloor \frac{k}{t}\rfloor t=r$. From the analysis of first part, we know the minimal area is achieved if $g$ is maximal. Thus, we only consider the permutation of $\mathbf{u}=[r,t,\dots,t]_{1\times g}$.  
For $d\beta \ge t\alpha$ and $r\alpha$ is the smallest area that contribute to $min(g(i),f(i))$, there is only one possibility for minimizing bandwidth if there exist $Y=\left\lbrace u_i| t\alpha \ge g(i) \right\rbrace$. Thus, let $\min \mathcal{C}_0 =\min (r\alpha,d\beta)=r\alpha$, $\widehat{C}$ can be minimized as the change of $\alpha$. In this way, $\inf f(i)=(d'-it)\beta=(d-r-it)\beta, i=1,\dots,g-1$.
\end{proof}
The following examples present the minimal capacity: 
\begin{example}
Consider \[n= 14,k=6,d=10,t=3,\] 
$C*= \min(3\alpha,d\beta) + \min(3\alpha,(d-t)\beta)$
\begin{numcases}
{C^*=}
6\alpha, & $\alpha \in (0,\frac{7}{3}\beta]$ \notag \\
3\alpha + 7\beta, & $\alpha \in (\frac{7}{3}\beta,\frac{10}{3}\beta]$ \notag \\
17\beta & $\alpha \in (\frac{10}{3}\beta, +\infty)$ \notag 
\end{numcases}
\end{example}
\begin{example}
Consider \[n= 14,k=7,d=10,t=3,\]
$C*= \min(\alpha,d\beta) + \min(3\alpha,(d-1)\beta)+\min(3\alpha,(d-t-1)\beta)$
\begin{numcases}
{C^*=}
7\alpha, & $\alpha \in (0,\frac{6}{3}\beta]$ \notag \\
4\alpha + 6\beta, & $\alpha \in (\frac{6}{3}\beta,\frac{9}{3}\beta]$ \notag \\
\alpha + 15\beta, & $\alpha \in (\frac{9}{3}\beta,10\beta]$ \notag \\
21\beta, & $\alpha \in (10\beta, +\infty) $ \notag 
\end{numcases}
\end{example}

To simplify notation, introduce $b_i = \inf f(i)$, we have the general form:
\begin{equation}
\mathcal{C}_* =
\begin{cases}
k\alpha & \alpha \in [0,\frac{b_{g-1}}{t}]\\
(k-t)\alpha + b_{g-1} & \alpha \in (\frac{b_{g-1}}{t},\frac{b_{g-2}}{t}] \\
\vdots & \vdots \\
(k-(g-1)t)\alpha + \sum_{j=1}^{g-1} b_j & \alpha \in (\frac{b_{1}}{t},\frac{b_0}{k-(g-1)t}]\\ 
b_0 + b_1 + \dots + b_{g-1} & \alpha \in (\frac{b_0}{k-(g-1)t},+\infty)
\end{cases}
\end{equation}
Suppose the minimal $\alpha^*$ such that $\mathcal{C}_* = B$ and $B \le b_0 + b_1 + \dots + b_{g-1}$, then 
\begin{equation}
\alpha_* =
\begin{cases}
\frac{B}{k} & B \in [0,\frac{kb_{g-1}}{t}]\\
\frac{B- b_{g-1}}{k-t} & B \in (\frac{kb_{g-1}}{t},\frac{(k-t)b_{g-2}}{t}+b_{g-1}] \\
\vdots & \vdots \\
\frac{B-\sum_{j=1}^{g-1}b_j}{k-(g-1)t} & B \in (\frac{(k-(g-1)t)b_1}{t}+\sum_{j=1}^{g-1}b_j, \sum_{j=0}^{g-1}b_j]
\end{cases}
\end{equation} 
Let $\beta = \frac{\gamma}{d}$, for $i=1,\dots, g-1$, 
\begin{equation}
\begin{split}
\sum_{j=g-i}^{g-1}b_j
& = \sum_{j=g-i}^{g-1} (d-jt)\beta \\
& = [i - \frac{(2g-i-1)it}{2}]\frac{\gamma}{d}\\
& = \frac{[2d-(2g-i-1)t]i}{2d}\gamma\\
& = p(i)\gamma
\end{split}
\end{equation}
and
\begin{equation}
\begin{split}
\frac{(k-it)b_{g-i}}{t}+\sum_{j=g-i}^{g-1}b_j
&=\frac{(k-it)(d-(g-i)t)}{t}\beta+p(i)\gamma\\
& = \frac{2k[d-(g-i)t]-i(i-1)t^2}{2td}\gamma\\
& = q(i)\gamma
\end{split}
\end{equation}
Thus we have the following storage-bandwidth tradeoff:
\begin{equation}
\alpha_* =
\begin{cases}
\frac{B}{k} & \gamma \in  [\frac{B}{q(i)},+\infty), \\
\frac{B-p(i)\gamma}{k-it} & \gamma \in [\frac{B}{q(i+1)}, \frac{B}{q(i)}).
\end{cases}
\end{equation}
where
\begin{equation}
p(i)=\frac{[2d-(2g-i-1)t]i}{2d},
\end{equation}
\begin{equation} \label{qi}
q(i)=\frac{2k[d-(g-i)t]-i(i-1)t^2}{2td}.
\end{equation}
\subsection{The Two Extreme Points} \label{3c}
If $t \mid k$, then $k=gt$ while if $t \nmid k$, let $k = (g-1)t+r$ where $0 < r < t$. Then,
\begin{equation}
q(i)
\begin{cases}
\frac{k(d-k+t)}{td} & t \mid k, i=1 \\
\frac{k(d'-k+r+t)}{td} & t \nmid k, i=2
\end{cases} 
\end{equation} 
\begin{figure*}[t] \centering
\includegraphics[width=6in]{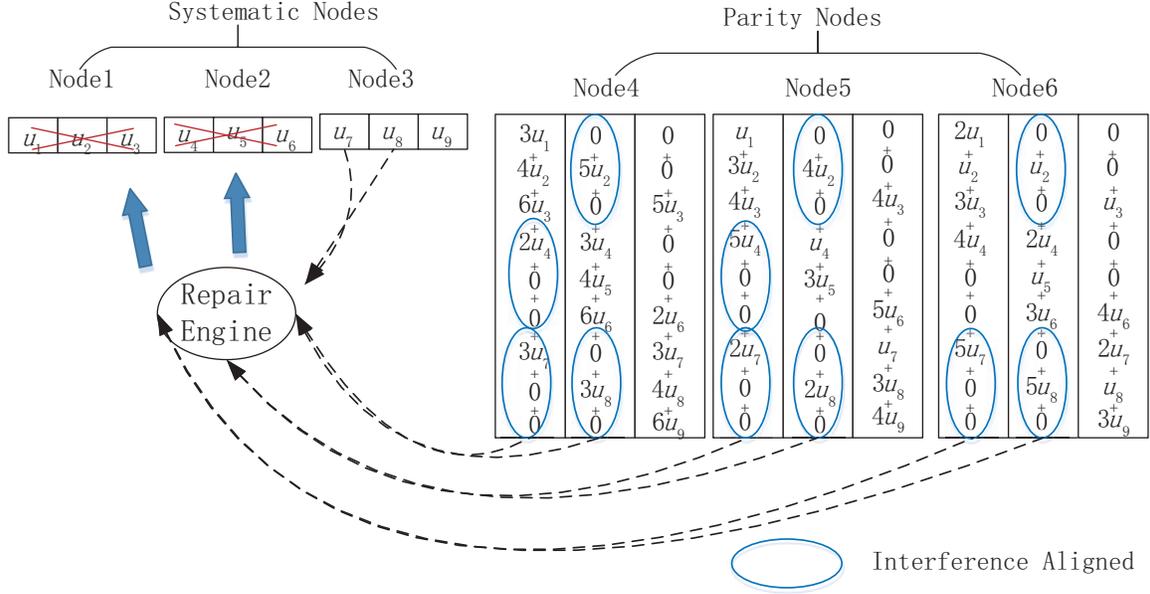}
\caption{Example of concurrent repair of [B=9,n=6,k=3,d=4,t=2] code over $\mathbb{F}_7$, the single repair example appears in literature \cite{6096412}.}
\label{fig_InterferenceAligned}
\end{figure*} 
We can obtain the minimal storage point:
\begin{equation} \label{minMS}
(\alpha_{MS}, \gamma_{MS})=(\frac{B}{k}, \frac{Btd}{k(d-k+t)})
\end{equation}
when i=g-1,
\begin{equation}
q(i+1)=q(g) = 
\begin{cases}
\frac{k(2d-k+t)}{2td} & t \mid k, \\
\frac{k(2d'-k+2r-t)}{2td} & t \nmid k
\end{cases}
\end{equation} 
We can obtain the minimal bandwidth point:
\begin{equation} \label{minMB}
(\alpha_{MB}, \gamma_{MB})=(\frac{2Bd}{k(2d-k+t)}, \frac{2Btd}{k(2d-k+t)})
\end{equation}
\subsection{Discussion and Simple Code Construction } \label{3d}
As can be seen from above expression (\ref{Di}) and (\ref{qi}), if $k=gt$ we have
\begin{equation}
q(i) = \frac{D_{i-1}}{td},
\end{equation}
implying that the bandwidth of concurrent repair mechanism have the same bound as that of the first phase of cooperative regenerating code, namely, $t\beta_1 = \beta$. 

 Although the total bandwidth of cooperative repair is minimized by both $\beta_1$ and $\beta_2$, it is impossible to maintain the tradeoff of cooperative RGC by further reducing $\beta_1$ and increasing $\beta_2$. Exactly, concurrent repair provides a shared scheme that can be used in non-fully distributed system.

The above mentioned constructions of MSR and MBR can be easily extended for this kind of multiple repair.

As an example of concurrent MSR depicted in figure (\ref{fig_InterferenceAligned}), we can extend interference alignment into multiple failure scenario, especially for exactly repair where $n,k, d+t \ge 2k-1$. Take [B=9,n=6,k=3,d=4,t=2] as an example, we can see the repair engine downloads $\frac{Bt}{k(d-k+t)}=2$ packets from $d=4$ nodes and recover the systematic node 1 and 2. That is, by eliminating $u_7$ and $u_8$ under the help of node 3, we can obtain the nonsingular coefficient matrix of the 6 variates,
\[A = \begin{bmatrix}
3 & 4 & 6 & 2 & 0 & 0 \\
0 & 5 & 0 & 3 & 4 & 6 \\
1 & 3 & 4 & 5 & 0 & 0 \\
0 & 4 & 0 & 1 & 3 & 4 \\
2 & 1 & 3 & 4 & 0 & 0 \\
0 & 1 & 0 & 2 & 1 & 3
\end{bmatrix}.\]
Then we can recover the 6 variates by using its inverse matrix. For more detail, please refer to \cite{6096412}. The example verifies the result of literature \cite{6620465} where $n=d+t=2k, k\ge 3, t=2$. The deterministic codes can be found in \cite{6847953}. 
By setting $d+t$ as a constant $n$, \cite{5978920} propose an adaptive regenerating codes. As $t$ changing, the number of helper $n-t$ is dynamically adjusted to adapt to the current state of the system. 

For minimal bandwidth point, \cite{5978920} figures out that fixed $t$ and $d$ are meaningful for constructing the related codes. Literature \cite{6566803} shows such explicit codes for $n \ge d+t, d \ge k, t \ge 1$. Comparing with single
repair model, the multiple mechanism requires smaller $\alpha$ result in lower storage cost at the minimal bandwidth point \cite{6566803}. For example, if $n=19,k=10,t=6,d=13$, then $\alpha=\frac{2dB}{k(2d-k+t)}=0.118B$, while if $n=19,k=10,t=1,d=18$ then $\alpha=\frac{2dB}{k(2d-k+1)}=0.133B$. It can save $12.7\%$ storage space. Similarly, the laze repair mechanism of Total Recall \cite{bhagwan2004total} can be introduced into minimal bandwidth repair as long as we set the threshold equal to proper $t$ \cite{5978920}. The repetitious details need not be elaborated here.
%
\subsection{Scalable-MDS Property} \label{3e}
Yuchong Hu \textit{et al.} \cite{5402494} have proposed the concept of (n,k) strong-MDS codes. In detail, a file is divided into $k(n-k)$ packets and encoded into $n(n-k)$ packets. Any $k(n-k)$ packets out of the $n(n-k)$ packets can reconstruct the original file. Namely, let $h_i$ denote the number of packets download from node $i \in \lbrace 1, \dots, n \rbrace$. For $0 \le h_i \le n-k$, we can reconstruct the original file if $\sum_{i=1}^{n}h_i=k(n-k)$.

Reconsider our generating model we can find that the whole bandwidth to generate $t=k$ nodes is just $B$ and the data are downloaded from $d(\ge k)$ nodes with each of $\beta= \frac{B}{d}$.
Since these $t=k$ generated nodes keep the MDS properties, implying that we can reconstruct the original file by getting data not only from $k$ nodes but a flexible $d$ nodes. In this way, we call it the scalable MDS properties, which verifies strong-MDS properties \cite{5402494}. 

\section{Expansion of regenerating mechanism} \label{4}
\begin{figure}[!t]
\includegraphics[width=3in]{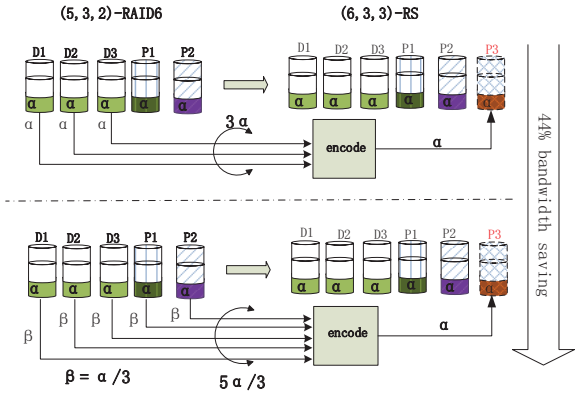}
\caption{Upgrade (5,3,2)-RAID6 to (6,3,3)-RS codes by using the idea of regenerating codes}
\label{fig_raidupgrade}
\end{figure}

\begin{figure}[!t]
	\includegraphics[width=3in]{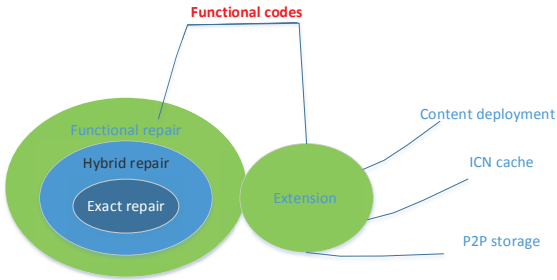}
	\caption{Expansion of function regenerating codes for scalable application}
	\label{fig_expansionfunc}
\end{figure}

In this section, a new application of regenerating codes is exploited.
Since the $(n,k)$-MDS property of erasure codes can provide ${n \choose k}$ choices for obtaining intact information in a distributed environment and dynamic $n$ can adjust the number of choices according to the popularity of certain information, a flexible system is taken into consideration in next generation networks which pay more attention to information itself, such as content delivery network (CDN) and information centric network (ICN).
Hence, regenerating codes can play another role in practical usage besides for archive or cold data storage. For instance depicted in figure \ref{fig_raidupgrade}, we add one parity node to transform RAID6 to RS by using the idea of regenerating codes, which can perform the same bandwidth saving feature as repair. This kind of upgrade adopts functional codes to guarantee the MDS property. In figure \ref{fig_expansionfunc}, we display the expansion. Based on the former analysis that the amount of data stored of minimal bandwidth point relates to $n$, we only consider the minimal storage point for practice.

\subsection{Upgrade (n,k) to (n+1,k)}

To add a new storage node to $(n,k)$ is to functionally repair for $(n+1,k)$ erasure codes, where the new node is treated as a virtual node of $(n+1,k)$ codes. 

As illustrated in figure \ref{fig_repairexpansion}, we represent the detail for the both repair and extension for $(n=5,k=3)$-MDS codes by using IA. Because the helper number $d$ is variable, $d_t=4$ for repair and $d_s=5$ for extension, we further divide each piece into $L$ segments by referring (\ref{MSR}), where $L = LCM(k(d_t-k+1),k(d_s-k+1))=LCM(6,9)=18$. Let $B=18$, then $\alpha=6$, $\beta_t=3$, $\beta_s=2$. Denote $\mathbf{a}=[a_1, \dots, a_{6}]^t$, $\mathbf{b}=[b_1, \dots, b_{6}]^t$, $\mathbf{c}=[c_1, \dots, c_{6}]^t$, where $[*]^t$ indicates a transpose. $\mathbf{A}_i$, $\mathbf{B}_i$ and $\mathbf{C}_i$ are the corresponding $\alpha \times \alpha$ generator submatrices (i=1,2,3). In the repair side, each helper provides $\beta_t=3$ linear combinations of the $\alpha = 6$ segments by multiplying $\alpha \times \beta_t$ repair project vectors $\mathbf{p}_j, j=1,2,3,4$, i.e., $\mathbf{a}^t\mathbf{p}_1$. In the extension side, each helper provides $\beta_s=2$ linear combinations of the $\alpha = 6$ segments by multiplying $\alpha \times \beta_d$ extension project vectors $\mathbf{q}_j, j=1,2,3,4$, i.e., $\mathbf{a}^t\mathbf{q}_1$.
\begin{figure}[h]
\includegraphics[width=3in]{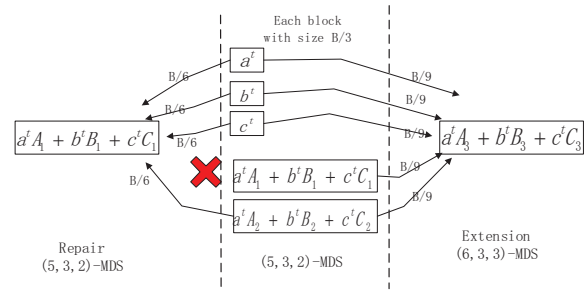}
\caption{The example for both repair and extension using regenerating code}
\label{fig_repairexpansion}
\end{figure}
\subsection{(n,k) to (n+s,k)}
%


So far we have demonstrated that we can download minimal data to construct one new node. Now we focus on whether is available to download more data from $d$ helpers to construct more nodes that have the same MDS property. We unify both repair and scalability in one model. Firstly, we prove the capability achievement of multiple node scenario. 
\subsubsection{the capability achievement of multiple node repair or upgrade }
Suppose the original file $\mathbf{M}$ can be denoted as a $B \times L$ matrix with each entry defined in finite field $\mathbb{F}$. That is, file $\mathbf{M}$ consists of $B$ strips, each with size $L$. For the seek of distributed storage, we divide $B$ into $k$ pieces, each with $\alpha$ strips, and then encode $k$ pieces into $n$ pieces, each with the same number of strips. In practice, we adapt $L$ to make strip as the minimal operation unit. Additionally, we denote $\mathcal{G}=\left\lbrace G_{1}, \dots, G_{n} \right\rbrace $ as the set of coding matrix in the same finite field, one element $G_{i}$ of which is a $B \times \alpha$ matrix and specified by $\alpha$ column vectors $\left\lbrace \mathbf{g}^{i}_{j}, j = 1, \dots, \alpha \right\rbrace$ with dimension $B$. It is used to obtain one coded piece by $\mathbf{M}^{T} G_{i}$.  
To keep MDS property, namely, to reconstruct the original file $\mathbf{M}$ from any $k$ out of $n$ nodes, say $x_{1}, \dots, x_{k}$, the span of the $k\alpha$ vectors in $\Phi = \left\lbrace G_{x_{1}}, \dots, G_{x_{k}} \right\rbrace$ should be full rank, which is 
\begin{equation} \label{MDSpro}
rank([ G_{x_{1}}, \dots, G_{x_{k}} ])=k*\alpha.
\end{equation}

Suppose a daemon connects to $d \ge k$  helpers, downloads $\beta$ linear combinations of $\alpha$ strips in each one and re-encode $\alpha$ strips for each of $r$ newcomers from the whole downloaded $d\beta$ strips. Let $\mathcal{G}^{new}=\left\lbrace G_{1}^{new}, \dots, G_{r}^{new} \right\rbrace$ as the set of new coding matrices, the whole process can be denoted as the following linear transformation, 
\begin{equation}
\begin{bmatrix}
G_{1}^{new} \\
\vdots \\
G_{r}^{new}
\end{bmatrix}
=[G_{x_{1}}P_{x_{1}}, \dots, G_{x_{d}}P_{x_{d}}]
\begin{bmatrix}
Z^{(1)}\\
\vdots \\
Z^{(r)}
\end{bmatrix}
\end{equation}
where $P_{x_{w}}$ is the project matrix of size $\alpha \times \beta$ on the $w$-th node out of $d$ nodes and $Z^{(\ell)}$ is the linear transform matrix of size $d\beta \times \alpha$ for the $\ell$-th newcomer. To guarantee MDS property is equivalent to make $k$ out of $n$ (if repair) or $n+s$ (if upgrade) encoding matrix satisfy (\ref{MDSpro}). Thus, a new span $\Phi^{new}$ containing $k\alpha$ linear independent vectors arbitrarily selected from $\mathcal{G}^{new}$ and $\mathcal{G}^{old}= \left\lbrace G_{1}, \dots, G_{T} \right\rbrace$ is constructed, where $T=n-t$ for repair or $T=n$ for upgrade. Then we have the following theorem:
\begin{theorem} \label{multipleMSR}
Given a scalable MDS code $(n, k, d)$, the minimal data downloaded from $d$ helpers to generating $r$ newcomers is  $\beta \ge \frac{rB}{d-k+r}$, whatever $r$ newcomers are replacement nodes for repair or added nodes for upgrade.
\end{theorem}
\begin{proof}
The problem exactly is to find the minimal $d\beta$ that makes the matrix $[\Phi^{new}]=[\mathcal{G}^{old}_{1,k-q}, \mathcal{G}^{new}_{1,q}]=[G_{y_1}, \dots, G_{y_{k-q}}, G_{x_1}^{new},\dots, G_{x_q}^{new}]$ be full rank, where $0 \le q \le r$.

In the beginning, we only take repair scenario into consideration. Without loss of generalization, we let $y_1=x_1,\dots,y_{k-q}=x_{k-q}$, then the vector format of $[\Phi^{new}]$ is $[\mathbf{g}^{x_1}_1, \dots, \mathbf{g}^{x_1}_\alpha; \dots;  \mathbf{g}^{x_{k-q}}_1, \dots, \mathbf{g}^{x_{k-q}}_\alpha; \mathbf{g}^{x^{new}_1}_1, \dots, \mathbf{g}^{x^{new}_1}_\alpha; \dots;$ $ \mathbf{g}^{x^{new}_q}_1, \dots, \mathbf{g}^{x^{new}_q}_\alpha]$. 
We denote $p_{i,j}^{x_w}$ as the $i$-th row and $j$-th column element of $P_{x_w}$ and $z^{\ell}_{i,j}$ as the $i$-th row and $j$-th column element of $Z^{(\ell)}$, then, 
\begin{equation} \label{onenodeprovide}
G_{x_w}P_{x_w}=
\begin{bmatrix}
\sum_{i=1}^{\alpha} p_{i1}^{x_w}\mathbf{g}^{x_w}_i & \ldots & \sum_{i=1}^{\alpha} p_{i\beta}^{x_w}\mathbf{g}^{x_w}_i
\end{bmatrix}.
\end{equation}
and the $h$-th vector of $G^{new}_{\ell}$ is
\begin{equation} \label{vector}
\begin{split}
\mathbf{g}^{x^{new}_{\ell}}_h 
&=\sum_{w=1}^{d} \sum_{j=1}^{\beta} \sum_{i=1}^{\alpha} z^{\ell}_{j+(w-1)\beta,h}  p_{i,j}^{x_w} \mathbf{g}^{x_w}_i\\
&= \sum_{w=1}^{d} \sum_{i=1}^{\alpha} f_{i,w,h}^\ell \mathbf{g}^{x_w}_i
\end{split}
\end{equation}
where $f_{i,w,h}^\ell = \sum_{j=1}^{\beta} z^{\ell}_{j+(w-1)\beta,h}  p_{i,j}^{x_w}$, $1 \le w \le d$, $1 \le \ell \le m$ and $h \in \left\lbrace 1,\dots,\alpha \right\rbrace$. 
By using elementary column operations on the matrix $[\mathcal{G}^{old}_{1,k-q}, \mathcal{G}^{new}_{1,q}]$, we initially update $\mathbf{g}^{x^{new}_{\ell}}_h$ by 
\begin{equation} \label{newvector}
\mathbf{\bar{g}}^{x^{new}_{\ell}}_h= \sum_{w=k-q+1}^{k} \sum_{i=1}^{\alpha} f_{i,w,h}^{'\ell} \mathbf{g}^{x_w}_i.
\end{equation}
Due to the symmetrical effect of MDS property, we treat the vectors in $\Phi$ as the base vectors and other vectors of $G \notin \Phi$ are the linear combinations of the base vectors. 
Then each $\mathbf{g^{x_w}_i}$ in $[G_{x_{k+1}}, \dots, G_{x_{d}}]$ can be denoted as the linear combination of $k\alpha$ base vectors, $\mathbf{\bar{g}}^{x^{new}_{\ell}}_h$ can also be treated as the new linear combination of the $r\alpha$ different base vector from $[G_{x_{k-q+1}}, \dots, G_{x_{k}}]$. If we guarantee each element in $[\mathbf{g}^{x^{new}_1}_1, \dots, \mathbf{g}^{x^{new}_1}_\alpha; \dots; \mathbf{g}^{x^{new}_q}_1, \dots, \mathbf{g}^{x^{new}_q}_\alpha]$ contains one non-eliminated base vector by using elementary column operations, the matrix would be full rank. In such way, each of the left $\beta(d-k+q)$ strips derived from $[G_{x_{k-q+1}}, \dots, G_{x_{d}}]$, if it is useful, at least contributes one mutually different base vector, which implies $q\alpha \le \beta(d-k+q)$. 

For $\alpha=B/k$, we have $\beta \ge \frac{qB}{k(d-k+q)}$. 
To make it suitable to all the $ 1 \le q \le r$, we have 
\begin{equation} \label{minimalstore}
\beta \ge \frac{rB}{k(d-k+r)}.
\end{equation}
Then, we prove it for upgrade scenario. We need to prove that the coding matrix of the non-helpers keep the same MDS property with the coding matrix of the newcomers. Based on the same assumption that the code vector of each matrix of the non-helpers is the linear combination of the base vectors, it  obviously satisfies the full rank requirement if we adopt the checking step.
Hence the claim follows.
\end{proof}
It is consistent with the result of the concurrent minimal storage point. 
Moreover, we have to determine the value of coefficients of each base vector result from different $P_{x_w}$ and $Z^{(\ell)}$ to ensure such linear independence. Because the above description is based on the assumption that $\beta$ and $\alpha$ is positive integer, we should select proper $\beta$ for application. In cooperative minimal storage regenerating process each newcomer at least need one strip from each helper, which implys the minimal value of $\beta$ is $r$ ($\beta =r \beta_1$). Then we have $\alpha = d-k+r$ and $B=k(d-k+r)$ to construct such codes. In fact, it is the same way as the scalar codes proposed in \cite{6283044}. For vector codes (e.g., $\beta = 2r$), the strip may not be the minimal operational unit and the size of the matrix $P_{x_w}$ and $Z^{(\ell)}$ will be expanded.

Secondly, we show the scalable repair and upgrade for multiple nodes.

\subsubsection{scalable repair and upgrade for multiple nodes}
Review the adaptive codes for repairing when $d+t=n$. Then a file is divided into $M=k(n-k)$ packets and encoded into $n(n-k)$ packets $x_{1,1},\dots, x_{1,n-k}, \dots, x_{n,1}, \dots, x_{n,n-k}$ by multiplying $k(n-k) \times n(n-k)$ matrix. Each one of the $n$ nodes stores $\alpha=n-k$ packets. When there are $t$ node failures, we download $t$ linear combinations of $n-k$ packets from each of $d$ helpers.

Similarly, we can apply the codes to extension by set $d+s=n$, treat $s$ newcomers as $s$ failures waiting for repair. In this way, there would be $n-s$ healthy nodes not to the helpers, while we need to guarantee them and the newcomers maintain the MDS properties.

\begin{theorem}
Given a distributed storage system with (n,k)-MDS codes based on MSR scheme, where $B=k(d-k+1)$,$\alpha=B/k=d-k+1$. If $d$ is fixed and there exist $d$ available helpers, we at least need to download 
\begin{equation}
\gamma^{(r)}=\frac{Br(d-r+1)}{k(d-k+1)}=r(d-r+1)
\end{equation}
data from all the $d$ helpers to generate $r$ nodes simultaneous.
We use a $d$-dimension vector $\mathbf{h}$ to indicate the downloaded capacities from the helpers and $h_{i}$ denotes the downloaded capacities from $i$-th helper, which subjects to $\sum_{i=1}^{d}h_{i} = r(d-r+1)$, $1 \le h_{i} \le r$.
\end{theorem}
\begin{proof}
It is easy to prove the special case when $h_{1}= h{2} = \dots = r$, in which we let $d-k+1=d'-k+r$, where $d'$ is the number of helpers depicted in equation (\ref{minMS}) and $r$ is treated as the number of newcomers, then we have
\begin{equation}
\frac{Brd'}{k(d'-k+r)}=\frac{Br(d-r+1)}{k(d-k+1)}=r(d-r+1)
\end{equation}
For general case, we only prove that the combinations are identical to that of the special case. As can be seen from the proof of theorem \ref{multipleMSR}, all the combinations are based on the same base vectors and the matrix can be full as long as there are enough linear independent combinations. Then we can replace a combination of one helper with that of the other helper. It means that we can transfer any general case to special case, then the claim follows.
%
\end{proof}




\begin{example}
Suppose n=7, k=4, d= 5, r = 2, s=1, the file size is $B=(d-k+r)k=12$ packets. To repair two node, we need to download $\frac{Br}{k(d-k+r)}=2$ from each of the 5 helpers, totally 10 packets. To upgrade from (7,4) to (8,4), we can set $d_s=6$ and then need $\frac{d_sB}{k(d_s-k+1)}=6$ packets from 6 helpers. To upgrade to (9,4), we can need $\frac{2(d_s-2+1)B}{k(d_s-k+1)}=10$ packets from 6 helpers, where $\mathbf{h}=\left[ 2,2,2,2,1,1 \right]$. If we set $d_s=7$, $\mathbf{h}=\left[ 2,2,2,1,1,1,1 \right]$
\end{example}

Thus, we give the basic steps of the whole procedure for scalable storage. We set $d=d_{r}=n-t$

(1) File distribution
\begin{enumerate}
\item For proper $(n,k,d,t)$, the original file is divided into $k(d-k+t)$ packets and then encoded into $n(d-k+t)$ packets $\left\langle x_{1,1}, \dots, x_{1,d-k+t}; \dots; x_{n,1}, \dots, x_{n,d-k+t}\right\rangle $ with $x_{i,j}=\mathbf{M}^T\mathbf{g}_i^j$, $1 \le i \le n$, $1 \le j \le d+t-k$. Each node $X_i$ stores $d-k+t$ packets $\left\langle  x_{i,1}, \dots, x_{i,d-k+t} \right\rangle $.
\end{enumerate}

(2) Date repairing
\begin{enumerate}
\item Choose a set of $t$ nodes $Y_1,\dots,Y_t$ from idle nodes as replacement nodes and a set of $d$ nodes $X'_1,\dots,X'_d$ from the surviving nodes as helpers. 
\item Each helper $X'_i$ transmits $t$ encoded packets $\left\langle  \chi_{i,1} \dots \chi_{i,t}  \right\rangle $ = $\left[ x_{i,1}, \dots, x_{i,d-k+t} \right]*P_{i}$ to the daemon, where $P_{i}$ is a $(d-k+t)\times t $ coefficient matrix.
\item The daemon encode the accepted packet $\left\langle  \chi_{1,1} \dots \chi_{1,t};\dots; \chi_{d,1} \dots \chi_{d,t} \right\rangle $ into linear independent packets $\left\langle  y_{1,1}, \dots, y_{1,d-k+t}; \dots; y_{t,1}, \dots, y_{t,d-k+t} \right\rangle $ by separately multiplying different linear transform  matrices $Z_1,\dots, Z_t$ of size $dt \times (d-k+t)$ .
\item The daemon distributes the encoded packet $\left\langle  y_{i,1}, \dots, y_{i,d-k+t} \right\rangle $ to node $Y_i$ for $1 \le i \le t$.
\end{enumerate}

(3) Storage upgrade from (n,k) to (n+s,k)
\begin{enumerate}
\item Choose a set of $s$ nodes $S_1,\dots,S_s$ from idle nodes as upgrade nodes and a set of $d_{s}$ nodes $X'_1,\dots,X'_{d_{s}}$ from the surviving nodes as helpers.
\item Each helper $X'_i$ transmits $h_i$ encoded packets $\left\langle  \varsigma_{i,1} \dots \varsigma_{i,h_i}  \right\rangle $ = $\left[ x_{i,1}, \dots, x_{i,d-k+t} \right]*R_{i}$ to the daemon, where $R_{i}$ is a $(d-k+t) \times h_i$ coefficient matrix and $h_i$ is subject to $\sum_{i=1}^{d_s} h_i = s(d-s+1)$ and $\lfloor \frac{s(d-s+1)}{d_{s}} \rfloor \le h_i \le s$. 
\item The deamon encodes the accepted pakects $\left\langle  \varsigma_{1,1} \dots \varsigma_{1,h_1}; \dots; \varsigma_{d_{s},1} \dots \varsigma_{d_{s},h_{d_s}}  \right\rangle $ into linear independent packets $\left\langle  s_{1,1}, \dots, s_{1,d-k+t}; \dots; s_{t,1}, \dots, s_{t,d-k+t} \right\rangle $ by separately multiplying different linear transform matrix $Z_1,\dots, Z_s$ of size $s(d-s+1) \times (d-k+t)$.
\item The daemon distributes $\left\langle  s_{i,1}, \dots, s_{i,d-k+t} \right\rangle $ to $S_i$ for $1 \le i \le s$.
\end{enumerate}
We have the following properties:
\begin{enumerate}
\item[1] \textbf{Strong-MDS}: The data collector can reconstruct the original data by downloading minimal $\sum_{i=1}^{k^*}v_{i} = B$ data from any $k^*$ out of $n$ nodes, where $k \le k^* \le n$ and $1 \le v_{i} \le k$.
\item[2] \textbf{Multiple repair}: It can concurrently repair $t$ failures by downloading minimal $\frac{Btd_{r}}{k(d_{r}-k+t)}$ from $d_{r}$ helpers with $\frac{Bt}{k(d_{r}-k+t)}$ for each, where $k \le d_{r} \le n-t$ and $1 \le t \le k$.
\item[3] \textbf{Scalable upgrade}: It can upgrade $(n,k)$ codes to $(n+s,k)$ codes by treating added nodes as failure ones so that implement the multiple repair procedure or downloading minimal $\sum_{i=1}^{d_{s}}h_{i}=\frac{Bs(d_{s}-s+1)}{k(d_{s}-k+1)}$ from $d_{s}$ helpers, where $h_{i}$ denotes the contribution of $i$-th helper, $0 < h_i \le s$, $k \le d_{s} \le n$ and $1 \le s \le k$.
\end{enumerate}

Observe the later two properties, we can fix the number of strips in each piece, namely $d_{r}+t=d_{s}+1$ are preferred in practice. Thus, the DSS can be more flexible to adjust the amount of storage result from the popularity of data.

\section{Conclusion} \label{5}
In this paper, we review the concurrent regenerating codes for multiple-node repair and give close-form expressions for the storage-bandwidth bound by using cut-set bound analysis. Referring the cooperative regenerating codes, the codes can simplify the design process of storage systems, which is practical in non-purely distributed storage environment. We show that the existing constructions of single-repair MSR with IA can be easily reformed to concurrent MSR. When the number of multiple nodes needed to concurrently repair is $k$, we build a connection between strong-MDS property with the codes. Besides repair, we expand the application of functional codes for extension to provide scalability for erasure coded distributed storage systems. Apart from simply reforming repairable codes
to scalable codes, we propose both coupled and decoupled design for the extension, which can be flexibly used in the distributed network. 

\bibliographystyle{IEEEtran}
\bibliography{IEEEabrv,myinfocomref}

\begin{thebibliography}{10}
\providecommand{\url}[1]{#1}
\csname url@samestyle\endcsname
\providecommand{\newblock}{\relax}
\providecommand{\bibinfo}[2]{#2}
\providecommand{\BIBentrySTDinterwordspacing}{\spaceskip=0pt\relax}
\providecommand{\BIBentryALTinterwordstretchfactor}{4}
\providecommand{\BIBentryALTinterwordspacing}{\spaceskip=\fontdimen2\font plus
\BIBentryALTinterwordstretchfactor\fontdimen3\font minus
  \fontdimen4\font\relax}
\providecommand{\BIBforeignlanguage}[2]{{%
\expandafter\ifx\csname l@#1\endcsname\relax
\typeout{** WARNING: IEEEtran.bst: No hyphenation pattern has been}%
\typeout{** loaded for the language `#1'. Using the pattern for}%
\typeout{** the default language instead.}%
\else
\language=\csname l@#1\endcsname
\fi
#2}}
\providecommand{\BIBdecl}{\relax}
\BIBdecl

\bibitem{ghemawat2003google}
S.~Ghemawat, H.~Gobioff, and S.-T. Leung, ``The google file system,'' in
  \emph{ACM SIGOPS Operating Systems Review}, vol.~37, no.~5.\hskip 1em plus
  0.5em minus 0.4em\relax ACM, 2003, pp. 29--43.

\bibitem{borthakur2008hdfs}
D.~Borthakur, ``Hdfs architecture guide,'' \emph{HADOOP APACHE PROJECT
  http://hadoop. apache. org/common/docs/current/hdfs design. pdf}, 2008.

\bibitem{garfinkel2007evaluation}
S.~L. Garfinkel, ``An evaluation of amazon’s grid computing services: Ec2,
  s3, and sqs,'' in \emph{Center for}.\hskip 1em plus 0.5em minus 0.4em\relax
  Citeseer, 2007.

\bibitem{cleversafe2008paradigm}
A.~Cleversafe, ``Paradigm shift in digital assest storage,'' \emph{Cleversafe
  Whitepaper}, 2008.

\bibitem{whitehouse2007gfs2}
S.~Whitehouse, ``The {GFS}2 filesystem,'' in \emph{Proceedings of the Linux
  Symposium}.\hskip 1em plus 0.5em minus 0.4em\relax Citeseer, 2007, pp.
  253--259.

\bibitem{borthakur2010hdfs}
D.~Borthakur, R.~Schmidt, R.~Vadali, S.~Chen, and P.~Kling, ``{HDFS RAID},'' in
  \emph{Hadoop User Group Meeting}, 2010.

\bibitem{li2003linear}
S.-Y. Li, R.~W. Yeung, and N.~Cai, ``Linear network coding,'' \emph{Information
  Theory, IEEE Transactions on}, vol.~49, no.~2, pp. 371--381, 2003.

\bibitem{dimakis2007benefits}
A.~G. Dimakis, P.~B. Godfrey, M.~J. Wainwright, and K.~Ramchandran, ``The
  benefits of network coding for peer-to-peer storage systems,'' in \emph{Third
  Workshop on Network Coding, Theory, and Applications}, 2007.

\bibitem{dimakis2010network}
A.~G. Dimakis, P.~B. Godfrey, Y.~Wu, M.~J. Wainwright, and K.~Ramchandran,
  ``Network coding for distributed storage systems,'' \emph{Information Theory,
  IEEE Transactions on}, vol.~56, no.~9, pp. 4539--4551, 2010.

\bibitem{RashmiShah-75}
K.~V. Rashmi, N.~B. Shah, and P.~V. Kumar, ``Optimal exact-regenerating codes
  for distributed storage at the msr and mbr points via a product-matrix
  construction,'' \emph{Information Theory, IEEE Transactions on}, vol.~57,
  no.~8, pp. 5227--5239, 2011.

\bibitem{SuhRamchandran-78}
C.~Suh and K.~Ramchandran, ``Exact-repair mds code construction using
  interference alignment,'' \emph{Information Theory, IEEE Transactions on},
  vol.~57, no.~3, pp. 1425--1442, 2011.

\bibitem{wu2010existence}
Y.~Wu, ``Existence and construction of capacity-achieving network codes for
  distributed storage,'' \emph{Selected Areas in Communications, IEEE Journal
  on}, vol.~28, no.~2, pp. 277--288, 2010.

\bibitem{5773063}
------, ``A construction of systematic mds codes with minimum repair
  bandwidth,'' \emph{Information Theory, IEEE Transactions on}, vol.~57, no.~6,
  pp. 3738--3741, June 2011.

\bibitem{dimakis2011survey}
A.~G. Dimakis, K.~Ramchandran, Y.~Wu, and C.~Suh, ``A survey on network codes
  for distributed storage,'' \emph{Proceedings of the IEEE}, vol.~99, no.~3,
  pp. 476--489, 2011.

\bibitem{wang2010mfr}
X.~Wang, Y.~Xu, Y.~Hu, and K.~Ou, ``{MFR:} multi-loss flexible recovery in
  distributed storage systems,'' in \emph{Communications (ICC), 2010 IEEE
  International Conference on}.\hskip 1em plus 0.5em minus 0.4em\relax IEEE,
  2010, pp. 1--5.

\bibitem{5402494}
Y.~Hu, Y.~Xu, X.~Wang, C.~Zhan, and P.~Li, ``Cooperative recovery of
  distributed storage systems from multiple losses with network coding,''
  \emph{Selected Areas in Communications, IEEE Journal on}, vol.~28, no.~2, pp.
  268--276, February 2010.

\bibitem{5962548}
K.~Shum, ``Cooperative regenerating codes for distributed storage systems,'' in
  \emph{Communications (ICC), 2011 IEEE International Conference on}, June
  2011, pp. 1--5.

\bibitem{5978920}
A.-M. Kermarrec, N.~Le~Scouarnec, and G.~Straub, ``Repairing multiple failures
  with coordinated and adaptive regenerating codes,'' in \emph{Network Coding
  (NetCod), 2011 International Symposium on}, July 2011, pp. 1--6.

\bibitem{6565355}
K.~Shum and Y.~Hu, ``Cooperative regenerating codes,'' \emph{Information
  Theory, IEEE Transactions on}, vol.~59, no.~11, pp. 7229--7258, Nov 2013.

\bibitem{6566803}
A.~Wang and Z.~Zhang, ``Exact cooperative regenerating codes with
  minimum-repair-bandwidth for distributed storage,'' in \emph{INFOCOM, 2013
  Proceedings IEEE}, April 2013, pp. 400--404.

\bibitem{6620465}
J.~Chen and K.~Shum, ``Repairing multiple failures in the suh-ramchandran
  regenerating codes,'' in \emph{Information Theory Proceedings (ISIT), 2013
  IEEE International Symposium on}, July 2013, pp. 1441--1445.

\bibitem{6847953}
J.~Li and B.~Li, ``Cooperative repair with minimum-storage regenerating codes
  for distributed storage,'' in \emph{INFOCOM, 2014 Proceedings IEEE}, April
  2014, pp. 316--324.

\bibitem{DBLP:journals/corr/abs-1302-3344}
\BIBentryALTinterwordspacing
R.~Li, J.~Lin, and P.~P.~C. Lee, ``{CORE:} augmenting regenerating-coding-based
  recovery for single and concurrent failures in distributed storage systems,''
  \emph{CoRR}, vol. abs/1302.3344, 2013. [Online]. Available:
  \url{http://arxiv.org/abs/1302.3344}
\BIBentrySTDinterwordspacing

\bibitem{6880379}
R.~Li, J.~Lin, and P.~Lee, ``Enabling concurrent failure recovery for
  regenerating-coding-based storage systems: From theory to practice,''
  \emph{Computers, IEEE Transactions on}, vol.~PP, no.~99, pp. 1--1, 2014.

\bibitem{you2014repairing}
P.~You, Y.~Peng, Z.~Huang, and C.~Wang, ``Repairing multiple data losses by
  parallel max-min trees based on regenerating codes in distributed storage
  systems,'' in \emph{Algorithms and Architectures for Parallel
  Processing}.\hskip 1em plus 0.5em minus 0.4em\relax Springer, 2014, pp.
  325--338.

\bibitem{fan2009diskreduce}
B.~Fan, W.~Tantisiriroj, L.~Xiao, and G.~Gibson, ``Diskreduce: Raid for
  data-intensive scalable computing,'' in \emph{Proceedings of the 4th Annual
  Workshop on Petascale Data Storage}.\hskip 1em plus 0.5em minus 0.4em\relax
  ACM, 2009, pp. 6--10.

\bibitem{6563278}
G.~Xylomenos, C.~N. Ververidis, V.~A. Siris, N.~Fotiou, C.~Tsilopoulos,
  X.~Vasilakos, K.~V. Katsaros, and G.~C. Polyzos, ``A survey of
  information-centric networking research,'' \emph{IEEE Communications Surveys
  Tutorials}, vol.~16, no.~2, pp. 1024--1049, Second 2014.

\bibitem{6062413}
N.~Shah, K.~Rashmi, P.~Kumar, and K.~Ramchandran, ``Distributed storage codes
  with repair-by-transfer and nonachievability of interior points on the
  storage-bandwidth tradeoff,'' \emph{Information Theory, IEEE Transactions
  on}, vol.~58, no.~3, pp. 1837--1852, March 2012.

\bibitem{hu2011ncfs}
Y.~Hu, C.-M. Yu, Y.~K. Li, P.~P. Lee, and J.~C. Lui, ``Ncfs: On the
  practicality and extensibility of a network-coding-based distributed file
  system,'' in \emph{Network Coding (NetCod), 2011 International Symposium
  on}.\hskip 1em plus 0.5em minus 0.4em\relax IEEE, 2011, pp. 1--6.

\bibitem{hu2012nccloud}
Y.~Hu, H.~C. Chen, P.~P. Lee, and Y.~Tang, ``Nccloud: Applying network coding
  for the storage repair in a cloud-of-clouds,'' in \emph{USENIX FAST}, 2012.

\bibitem{6096412}
N.~Shah, K.~Rashmi, P.~Kumar, and K.~Ramchandran, ``Interference alignment in
  regenerating codes for distributed storage: Necessity and code
  constructions,'' \emph{Information Theory, IEEE Transactions on}, vol.~58,
  no.~4, pp. 2134--2158, April 2012.

\bibitem{bhagwan2004total}
R.~Bhagwan, K.~Tati, Y.~Cheng, S.~Savage, and G.~M. Voelker, ``Total recall:
  System support for automated availability management.'' in \emph{NSDI},
  vol.~4, 2004, pp. 25--25.

\bibitem{6283044}
N.~Le~Scouarnec, ``Exact scalar minimum storage coordinated regenerating
  codes,'' in \emph{Information Theory Proceedings (ISIT), 2012 IEEE
  International Symposium on}, July 2012, pp. 1197--1201.

\end{thebibliography}

\end{document}